\documentclass[10pt,a4paper]{article}

\usepackage[linesnumbered,ruled,vlined,noresetcount]{algorithm2e}
\DontPrintSemicolon
\SetArgSty{textnormal}
\SetKwProg{Fn}{function}{:}{end}
\SetKwProg{Slot}{In Slot}{:}{end}
\SetKwProg{pFor}{for each}{ do in parallel}{end}

\usepackage{multirow}
\usepackage{color}

\usepackage{authblk}
\usepackage{geometry}
\usepackage{amsmath,amssymb,amsthm}
\usepackage{graphicx}
\usepackage{hyperref}

\newtheorem{lemma}{Lemma}
\newtheorem{theorem}{Theorem}
\newtheorem{remark}{Remark}
\theoremstyle{definition}
\newtheorem*{psketch}{Proof Sketch}
\newtheorem*{note}{Note}
 \newtheorem{definition}{Definition}

\newcommand{\modulespace}{\vspace{0.3cm}}

\newcommand{\diam}{D}
\newcommand{\sinit}{s_{\mathrm{init}}}
\newcommand{\nextc}{\mathit{next}}
\newcommand{\aset}{A}

\newcommand{\ie}{{\em i.e.},~}
\newcommand{\eg}{{\em e.g.},~}
\newcommand{\idmax}{\mathrm{id}_{\mathrm{max}}}
\newcommand{\amax}{a_{\mathrm{max}}}
\newcommand{\alast}{a_{\mathrm{last}}}
\newcommand{\lctn}{\nu}
\newcommand{\id}{\mathtt{ID}}
\newcommand{\pin}{\mathtt{pin}}
\newcommand{\pout}{\mathtt{pout}}
\newcommand{\settled}{\mathtt{settled}}
\newcommand{\mode}{\mathtt{mode}}
\newcommand{\parent}{\mathtt{parent}}
\newcommand{\help}{\mathtt{help}}
\newcommand{\found}{\mathtt{found}}
\newcommand{\checked}{\mathtt{checked}}
\newcommand{\nxt}{\mathtt{next}}
\newcommand{\leader}{\mathtt{leader}}
\newcommand{\level}{\mathtt{level}}
\newcommand{\needInit}{\mathtt{InitProbe}}
\newcommand{\done}{\mathtt{done}}

\newcommand{\levelL}{\mathtt{level}_L}
\newcommand{\levelS}{\mathtt{level}_S}

\newcommand{\Probe}{\mathtt{Probe}}
\newcommand{\Settle}{\mathtt{Settle}}

\newcommand{\full}{1}
\newcommand{\emp}{0}

\newcommand{\CO}[1]{\texttt{/***************}\ #1\ \texttt{***************/}\;}

\newcommand{\tr}{\mathbf{true}}
\newcommand{\fl}{\mathbf{false}}

\newcommand{\rootal}{\mathbf{RootedDisp}}
\newcommand{\general}{\mathbf{GeneralDisp}}
\newcommand{\rootopt}{\mathbf{RootedOpt}}

\newcommand{\uemp}{U_{\emp}}

\newcommand{\ubot}{U_{\bot}}

\newcommand{\calA}{\mathcal{A}}
\newcommand{\calM}{\mathcal{M}}

\title{
Near-linear Time Dispersion of Mobile Agents
\thanks{
This work is partly supported by JSPS KAKENHI
20H04140, 
20KK0232, 
21K17706, 
22K11971, 
and 
23K28037, 
and JST FOREST Program JPMJFR226U. 
}
} 
\date{}

\author[1]{Yuichi Sudo\thanks{Corresponding Author: sudo@hosei.ac.jp}}
\author[2]{Masahiro Shibata}
\author[3]{Junya Nakamura}
\author[4]{Yonghwan Kim}
\author[5]{Toshimitsu~Masuzawa}

\affil[1]{Hosei University, Japan}
\affil[2]{Kyusyu Institute of Technology, Japan}
\affil[3]{Toyohashi University of Technology, Japan}
\affil[4]{Nagoya Institute of Technology, Japan}
\affil[5]{Osaka University, Japan}

\begin{document}

\maketitle


\begin{abstract}
Consider that there are $k\le n$ agents in a simple, connected, and undirected graph $G=(V,E)$ with $n$ nodes and $m$ edges. The goal of the dispersion problem is to move these $k$ agents to mutually distinct nodes. Agents can communicate only when they are at the same node, and no other communication means, such as whiteboards, are available. We assume that the agents operate synchronously. 
We consider two scenarios: when all agents are initially located at a single node (rooted setting) and when they are initially distributed over one or more nodes (general setting). Kshemkalyani and Sharma presented a dispersion algorithm for the general setting, which uses $O(m_k)$ time and $\log(k + \Delta)$ bits of memory per agent [OPODIS 2021], where $m_k$ is the maximum number of edges in any induced subgraph of $G$ with $k$ nodes, and $\Delta$ is the maximum degree of $G$. This algorithm is currently the fastest in the literature, as no $o(m_k)$-time algorithm has been discovered, even for the rooted setting.
In this paper, we present significantly faster algorithms for both the rooted and the general settings.
First, we present an algorithm for the rooted setting that solves the dispersion problem in $O(k\log \min(k,\Delta))=O(k\log k)$ time using $O(\log (k+\Delta))$ bits of memory per agent. Next, we propose an algorithm for the general setting that achieves dispersion in $O(k \log k \cdot \log \min(k,\Delta))=O(k \log^2 k)$ time using $O(\log (k+\Delta))$ bits.
Finally, for the rooted setting, we give a time-optimal (i.e.,~$O(k)$-time) algorithm with $O(\Delta+\log k)$ bits of space per agent. 
All algorithms presented in this paper work only in the synchronous setting, while several algorithms in the literature, including the one given by Kshemkalyani and Sharma at OPODIS 2021, work in the asynchronous setting.
\end{abstract}

\section{Introduction}
\label{sec:intro}
In this paper, we focus on the dispersion problem involving mobile entities, referred to as mobile agents, or simply, \emph{agents}. At the start of an execution, $k$ agents are arbitrarily positioned at nodes of an undirected graph $G = (V, E)$ with $n$ nodes and $m$ edges. The objective is to ensure that all agents are located at mutually distinct nodes. This problem was originally proposed by Augustine and Moses Jr.~\cite{AM18} in 2018.
A particularly intriguing aspect of this problem is the unique computation model. Unlike many other models involving mobile agents on graphs, we do not have access to node identifiers, nor can we use local memory at each node. In this setting, an agent cannot retrieve or store any information from or on a node when it visits. However, each of the $k$ agents possesses a unique identifier and can communicate with each other when they are at the same node in the graph. 
The agents must collaboratively solve the dispersion problem through this direct communication.

\begin{table}
\center
\caption{Dispersion of mobile agents on an arbitrary graph ($\tau = \min(k,\Delta)$)
The algorithm of \cite{KMS19} needs to know
an asymptotically tight upper bound on $m_k$ in advance. 
(Since $m_k \le \min(m,k\Delta,\binom{k}{2})$, knowing tight upper bounds on $m$, $k$, and $\Delta$ is sufficient, but it increases the running time to $O(\min(m,k\Delta,\binom{k}{2})\log k)$.)
The $\log k$ term can be eliminated from the space complexities marked with daggers ($\dagger$) if you choose to disregard the memory space required for each agent to store its own identifier. For example, the proposed algorithm mentioned in Theorem \ref{theorem:rootal} requires only $O(\log \Delta)$ bits per agent. 
The space complexity of the first algorithm given by \cite{KA19}, which is marked with a double dagger ($\ddagger$),
can be decreased to $O(k \log \Delta)$ if we assume that the number of possible agent-identifiers is $O(k)$.
All algorithms listed in this table are deterministic.
}
\label{table:existing}
\vspace{0.3cm}

\begin{tabular}{c c c c c c}
\hline
 & Memory per agent & Time & General/Rooted &  Async./Sync.\\
\hline
\cite{AM18}
& $O(\log (k+\Delta))$ $\dagger$ & $O(m_k)$&rooted&async.\\
\cite{KMS22}&
$O(\diam+\Delta \log k)$&$O(\diam \Delta(\diam+\Delta))$&rooted&async.\\
Theorem \ref{theorem:rootal}&
$O(\log (k+\Delta))$ $\dagger$
&$O(k \log \tau)$&rooted& sync.\\
Theorem \ref{theorem:rootopt}&
$O(\Delta + \log k)$ $\dagger$
&$O(k)$&rooted& sync.\\
\hline 
\cite{KA19}
& $O(k\log (k+\Delta))$ $\ddagger$ & $O(m_k)$&general& async.\\
\cite{KA19}& $O(\diam \log \Delta + \log k)$ $\dagger$
&$O(\Delta^\diam)$&general& async.\\
\cite{KA19}&
$O(\log (k+\Delta))$&$O(m_k\cdot k)$&general& async.\\
\cite{KMS19}&
$O(\log (k+\Delta) )$&$O(m_k \log k)$&general& sync.\\
\cite{SSKM20}&

$O(\log (k+\Delta) )$&$O(m_k \log k)$&general& sync.\\
\cite{KS21}&
$O(\log (k+\Delta) )$&$O(m_k)$&general&async.\\
Theorem \ref{theorem:general} & $O(\log (k+\Delta))$&$O(k \log^2 k)$&general& sync.\\
\hline
Lower bound & any & $\Omega(k)$ & any & any\\
\hline
\end{tabular}
\end{table}

Several algorithms have been introduced in the literature to solve the dispersion problem. 
This problem has been examined in two different contexts within the literature: the \emph{rooted setting} and the \emph{general setting}. In the rooted setting, all $k$ agents initially reside at a single node. On the other hand, the general setting imposes no restrictions on the initial placement of the $k$ agents.
For any $i \in [1,n]$, we define $m_i$ as the maximum number of edges in any $i$-node induced subgraph of $G$. The parameter $m_k$, where $k$ is the number of agents, serves as an upper bound on the number of edges connecting two nodes, each hosting at least one agent, in any configuration. Consequently, $m_k$ frequently appears in the time complexities of dispersion algorithms.
This is because (i) solving the dispersion problem essentially requires finding $k$ distinct nodes, and (ii) the simple depth-first search (DFS), employed as a submodule by many dispersion algorithms, needs to explore $m_k$ edges to find $k$ nodes.

Table \ref{table:existing} provides a summary of various dispersion algorithms found in the literature, all designed for arbitrary graphs. Here, $\Delta$ and $D$ are the maximum degree and the diameter of a graph, respectively, and $\tau = \min(k,\Delta)$.
Augustine and Moses Jr.~\cite{AM18} introduced a simple algorithm, based on depth-first search (DFS), for the rooted setting. This algorithm solves the dispersion problem in $O(m_k)$ time using $O(\log \Delta)$ bits of space per agent. 
Kshemkalyani and Ali~\cite{KA19} provides two algorithms that accomplish dispersion in the general setting:
an $O(m_k)$-time and $O(k \log (k+\Delta))$-space algorithm, and an $O(m_k \cdot k)$-time and $O(\log (k+\Delta))$-space algorithm, offering a trade-off between time and space.
The first is faster but needs more space, while the second is slower but more memory-efficient.
Kshemkalyani, Molla, and Sharma~\cite{KMS19} found a middle ground with an algorithm that runs in $O(m_k \log k)$ time and uses $O(\log (k+\Delta))$ bits of each agent's memory. This algorithm, however, requires a priori global knowledge, asymptotically tight upper bounds on $m_k$, to attain its time upper bound.
Shintaku, Sudo, Kakugawa, and Masuzawa~\cite{SSKM20} managed to eliminate this requirement for global knowledge. More recently, Kshemkalyani and Sharma~\cite{KS21} removed the $\log k$ factor from the running time. This algorithm also works in an asynchronous setting, meaning the agents do not need to share a common clock.
Any dispersion algorithm requires at least $\Omega(k)$ time, which is almost trivial, but we will provide a proof for completeness in this paper. No other lower bounds on the time complexity of dispersion have been established in the literature.
Thus, there is still a significant gap between the best known upper bound $O(m_k)$ and this lower bound of $\Omega(k)$ because $m_k = \Theta(k^2)$ holds in many graph classes.
Note that $m_k = \Theta(k^2)$ may hold even in a sparse graph when $k = O(\sqrt{n})$. 



All the algorithms mentioned above are based on DFS. However, a few algorithms \cite{KA19, KMS22} are designed based on BFS (breadth-first search) and exhibit different performance characteristics. Notably, their upper bounds on running time do not depend on the number of agents $k$, but depend on diameter $\diam$ and the maximum degree $\Delta$ of a graph.

\begin{note}[Space Complexity]
\label{note:space}
Conforming to the convention in the studies of mobile agents \cite{Das19}, this paper, including Table \ref{table:existing}, evaluates the space complexity of an algorithm as the maximum size of \emph{persistent memory} needed by an agent during its execution. Persistent memory refers to the information an agent carries when it moves from one node to another and does not include the \emph{working memory} used for local computations at nodes. This persistent memory includes the space required to store its own identifier. Since the $k$ agents are labeled with unique identifiers, every algorithm requires $O(\log k)$ bits per agent.
\end{note}

\begin{note}[Parameter $m_k$]
\label{note:m_k}
The parameter $m_k$ is introduced in this paper and has not been previously utilized in the literature. Traditionally, the running times of DFS-based algorithms are represented using the parameter $\min(m, k \Delta)$ or $\min(m, k \Delta, \binom{k}{2})$, which are always greater than or equal to $m_k$. The parameter $m_k$ may be better to represent them because there are some graph classes where $m_k = o(k^2)$ while $\min(m, k \Delta, \binom{k}{2}) = \Omega(k^2)$.
\end{note}

\subsection{Our Contribution}
In this paper, we drastically reduce the gap between the upper bound $O(m_k)$ and the lower bound $\Omega(k)$ mentioned earlier. As previously noted, $m_k = \Theta(k^2)$ may hold even in sparse graphs, making this gap significant. Let $\tau = \min(k, \Delta)$. We present two algorithms: one for the rooted setting that achieves dispersion in $O(k \log \tau) = O(k \log k)$ time using $O(\log \Delta)$ bits, and the other for the general setting that achieves dispersion in $O(k \log k \cdot \log \tau) = O(k \log^2 k)$ time using $O(\log (k + \Delta))$ bits. The upper bounds obtained here match the lower bounds in both the rooted and the general settings when ignoring poly-logarithmic factors.

To achieve this upper bound, we introduce a new technique. Like many existing algorithms, our algorithms are based on Depth-First Search (DFS). That is, we let agents run DFS on a graph and place or \emph{settle} an agent at each unvisited node they find.
Each time unsettled agents find an unvisited node $v$, one of the agents settles at $v$, and the others try to find
an unvisited neighbor of $v$. If such a neighbor exists, they move to it. If no such neighbor exists, they go back to the parent of $v$ in the DFS tree. To find an unvisited neighbor, all DFS-based dispersion algorithms in the literature make the unsettled agents visit those neighbors sequentially, \ie one by one. This process obviously requires $\Omega(\tau)$ time.
We break this barrier and find an unvisited neighbor of $v$ in $O(\log \tau)=O(\log k)$ time,
with the help of the agents already settled at neighbors of the current location $v$.

Our goal here is to find any one unvisited neighbor of $v$ if it exists, not to find all of them. Consider the case where there are only two agents $a$ and $b$ at $v$, $a$ is settled at $v$, and $b$ is still an unsettled agent. Agent $b$ visits a neighbor of $v$ and if $b$ finds a settled agent at that node, $b$ brings that agent to $v$. Consequently, there are two agents on $v$, excluding $a$, so we can use these two to visit two neighbors of $v$ in parallel. Again, if there are settled agents on both nodes, those agents will be brought to $v$. Importantly, the number of agents at $v$, excluding $a$, doubles each time this process is repeated until an unvisited neighbor is found.
Therefore, over time, we can check neighbors of $v$ in parallel with an exponentially increasing number of agents. As a result, we can finish this search or \emph{probing} process in $O(\log \tau)$ time. Thereafter, we allow the helping agents we brought to $v$ to return to their original nodes, or their \emph{homes}. Since we perform the probing process only $O(k)$ times in total throughout DFS, a simple analysis shows that dispersion can be achieved in $O(k \log \tau)$ time in the rooted setting. We call the resulting DFS the \emph{HEO (Helping Each Other)-DFS} in this paper.

In the general setting, like in existing studies, we conduct multiple DFSs in parallel, each starting from a different node. While the DFS performed in existing research requires $\Theta(m_k)$ time, we use HEO-DFS, thus each DFS completes in $O(k \log \tau)$ time. Thus, at first glance, it seems that dispersion can be achieved in $O(k \log \tau)$ time. However, this analysis does not work so simply because each DFS interferes with each other. Our proposed algorithm employs the method devised by Shintaku et al.~\cite{SSKM20} to efficiently merge multiple DFSs and run HEO-DFSs in parallel with this method. The merge process incurs an $O(\log k)$ overhead, so we solve the dispersion problem in $O(k (\log k)\cdot (\log \tau))=O(k \log^2 k)$ time.

It might seem that the overhead can be eliminated by using the DFS parallelization method
proposed by Kshemkalyani and Sharma \cite{KS21}, instead of the method of Shintaku et al.~\cite{SSKM20}.
However, this is not the case because our HEO-DFS is not compatible with the parallelization method of Kshemkalyani and Sharma. Specifically, their method entails a process such that one DFS absorbs another when multiple DFSs collide. During this process, it is necessary to gather the agents in the absorbed side to a single node, which requires $\Theta(m_k)$ time. Our speed-up idea effectively works for finding an unvisited neighbor, but it does not work for the acceleration of gathering agents dispersed on multiple nodes. Therefore, it is unlikely that our HEO-DFS can be combined with the method of Kshemkalyani and Sharma.

The two algorithms mentioned above are nearly time-optimal, \ie requiring $O(k \cdot \log^c k)$ time for some constant $c$. We also demonstrate that in the rooted setting, a time-optimal algorithm based on the HEO-DFS can be achieved if significantly more space is available, specifically $O(\Delta)$ bits per agent.

To the best of our knowledge, HEO-DFS is a novel approach, and no similar techniques have been used in the literature concerning mobile agents and mobile robots. While we demonstrate that HEO-DFS significantly reduces the running time of dispersion algorithms, this technique may also prove useful for addressing other fundamental problems such as exploration and gathering.

A drawback of our HEO-DFS is that it fundamentally requires a synchronous model, even in the rooted setting; \ie it does not function in an asynchronous model. In HEO-DFS, we attempt to find an unvisited neighbor of the current location with the help of agents settled on other neighbors. These agents must return to their homes once the probing process is completed. In an asynchronous model, unsettled agents (and/or helping agents) may visit those homes before their owners return, disrupting the consistent behavior of HEO-DFS. Therefore, the algorithm provided by Kshemkalyani and Sharma \cite{KS21} remains the fastest for the asynchronous model. It is still an open question whether there exists a $o(k^2)$-time algorithm that accommodates asynchronicity.

\begin{note}[Termination]
In this paper, we do not explicitly mention how the agents terminate the execution of a given algorithm.
In many cases,
termination is straightforward without any additional assumptions in the rooted setting, while in the general setting, additional assumptions are required. Specifically, in the general setting, all algorithms listed in Table \ref{table:existing}, except for the $O(\Delta^d)$-time algorithm presented in \cite{KMS20} \footnote{
However, in this algorithm, the agents do not terminate simultaneously, and they require the ability to detect whether or not there is a terminated agent at the current location.
}, require both a synchronous setting and global knowledge such as (asymptotically tight upper bounds on) $m_k$ and $k$. With these assumptions, the agents can easily terminate simultaneously after a sufficiently large number of steps, \eg $\Theta(k \log^2 k)$ steps in our algorithm for the general setting.
Thus, when termination is required, our general setting algorithm no longer exhibits disadvantages compared to existing algorithms: all existing algorithms, except for the $O(\Delta^D)$-time one~\cite{KMS20}, also require a synchronous setting (and some global knowledge).

 For completeness, we present how the agents terminate in our algorithms for the rooted setting, which is almost trivial, 
 in Section \ref{sec:termination} in the appendix.
\end{note}

\subsection{Further Related Work}
The dispersion problem has been studied not only for arbitrary undirected graphs but also for graphs with restricted topologies such as trees \cite{AM18}, grids \cite{KMS19,KMS20b}, and dynamic rings \cite{AAM+18}. Additionally, several studies have explored randomized algorithms to minimize the space complexity of dispersion \cite{MM19,DBS23}, and others have focused on fault-tolerant dispersion \cite{MMM21,CKM+23}. Kshemkalyani et al.~\cite{KMS22} introduced the \emph{global communication model}, where all agents can communicate with each other regardless of their locations. In contrast, the standard model, where
only the agents co-located at the same node can communicate with each other, is sometimes referred to as the \emph{local communication model}. All algorithms listed in Table \ref{table:existing} assume the local communication model and are deterministic. 

Exploration by a single mobile agent is closely related to the dispersion problem. The exploration problem requires an agent to visit all nodes of a graph. Many studies have addressed the exploration problem, and numerous efficient algorithms, both in terms of time and space, have been presented in the literature \cite{PDD+96,PP99,SBN15,SOK24}. In contrast to exploration, the dispersion problem only requires finding $k$ nodes, and we can use $k$ agents to achieve this. Our HEO-DFS take advantage of these differences to solve the dispersion problem efficiently.

\subsection{Organization of this paper}

The remainder of this paper is structured as follows:
\begin{itemize}
    \item Section \ref{sec:model} introduces the basic definitions, the model of computation, and the formal problem specification.
    \item Sections \ref{sec:rooted} and \ref{sec:general} present near time-optimal dispersion algorithms in the rooted and general settings, respectively, utilizing $O(\log (k+\Delta))$ bits per agent.
    \item Section \ref{sec:rootopt} describes a time-optimal algorithm for the rooted setting, achieved by increasing memory usage to $O(\Delta + \log k)$ bits per agent, trading off increased memory for optimal time performance.
    \item Section \ref{sec:discussion} summarizes the contributions of this work and discusses the conditions under which the HEO-DFS technique can accommodate the lack of synchronicity.
    \item The Appendix details the termination procedures for the two proposed algorithms designed for the rooted setting.
\end{itemize}

\section{Preliminaries}
\label{sec:model}
Let $G=(V,E)$ be any simple, undirected, and connected graph.
Let $n=|V|$ and $m=|E|$. 
We denote the set of \emph{neighbors} of node $v \in V$ 
by $N(v) = \{u \in V \mid \{u,v\} \in E\}$
and the degree of a node $v$ by $\delta_v = |N(v)|$.
Let $\Delta = \max_{v \in V} \delta_v$,
\ie $\Delta$ is the maximum degree of $G$.
The nodes are anonymous, \ie they do not have unique identifiers. However, the edges incident to a node $v$ 
are locally labeled at $v$
so that an agent located at $v$ can distinguish those edges.
Specifically, those edges have distinct labels
$0,1,\dots,\delta_v-1$ at node $v$.
We call these local labels \emph{port numbers}. 
We denote the port number assigned at $v$ for edge $\{v,u\}$
by $p_v(u)$.
Each edge $\{v,u\}$ has two endpoints, thus has labels $p_{u}(v)$ and $p_{v}(u)$.
Note that these labels are independent,
\ie $p_{u}(v) \neq p_{v}(u)$ may hold.
For any $v \in V$,
we define $N(v,i)$ as the node $u \in N(v)$ such that $p_v(u)=i$.
For simplicity, we define $N(v,\bot) = v$ for all $v \in V$.

We consider that $k$ agents exist in graph $G$, where $k \le n$.
The set of all agents is denoted by
$\aset$.
Each agent is always located at some node in $G$,
\ie the move of an agent is \emph{atomic}
and an agent is never located at an edge at any time step
(or just \emph{step}).
The agents have unique identifiers,
\ie each agent $a$ has a positive integer as its identifier $a.\id$
such that $a.\id \neq b.\id$
for any $b \in \aset \setminus \{a\}$.
The agents know a common upper bound $\idmax \ge \max_{a \in \aset}a.\id$
such that $\idmax = \mathrm{poly}(k)$,
thus the agents can store the identifier of any agent
on $O(\log k)$ space.
Each agent has a read-only variable $a.\pin \in \{0,1,\dots,\Delta-1\}\cup\{\bot\}$.
At time step $0$, $a.\pin = \bot$ holds.
For any $t \ge 1$,
if $a$ moves from $u$ to $v$ at step $t-1$,
$a.\pin$ is set to $p_v(u)$ (or the port of $v$ incoming from $u$)
at the beginning of step $t$.
If $a$ does not move at step $t-1$,
$a.\pin$ is set to $\bot$.
We call the value of $a.\pin$ the incoming port of $a$.
The values of all variables in agent $a$,
excluding its identifier $a.\id$ and special variables $a.\pin, a.\pout$,
constitute the state of $a$.
(We will see what is $a.\pout$ later.)

The agents are synchronous
and are given a common algorithm $\calA$.
An algorithm $\calA$ must specify the initial state $\sinit$ of agents.
All agents are in state $\sinit$ at time step $0$.
Let $\aset(v,t) \subseteq \aset$ denote the set of agents located at node $v$ at time step $t \ge 0$.
At each time step $t \ge 0$,
each agent $a \in \aset(v,t)$ is given
the following information as the inputs:
(i) the degree of $v$, 
(ii) its identifier $a.\id$, 
and 
(iii) a sequence of triples $((b.\id, s_b, b.\pin))_{b \in \aset(v,t)}$,
where $s_b$ is the current state of $b$.
Note that each $a \in \aset(v,t)$ can obtain its current state $s_a$ and $a.\pin$
from the sequence of triples since $a$ is given its ID as the second information.
Then, it updates the variables in its memory space in step $t$,
including a variable
$a.\pout \in \{\bot,0,1,\dots,\delta_v-1\}$,
according to algorithm $\calA$.
Finally, each agent $a \in \aset(v,t)$
moves to node $N(v,a.\pout)$.
Since we defined $N(v,\bot)=v$ above,
agent $a$ with $a.\pout = \bot$ stays in $v$ in step $t$.

A node does not have any local memory
accessible by the agents. 
Thus, the agents can coordinate only by communicating with the co-located agents. 
No agents are given any global knowledge 
such as $m$, $\Delta$, $k$, and $m_k$ in advance.

A function $C:\aset \to \calM_{\calA} \times V \times \{\bot,0,1,\dots,\Delta-1\}$
is called a global state of the network or a \emph{configuration}
if $C(a) = (s,v,q)$ yields $q = \bot$ or $q < \delta_v$ for any $a \in A$,
where $\calM_{\calA}$ is the (possibly infinite) set of all agent-states.
A configuration specifies
the state, location, and incoming port
of each $a \in A$.
In this paper, we consider only deterministic algorithms.
Thus, if the network is in a configuration $C$ at a time step $t$, 
a configuration $C'$ in the next step $t+1$ is uniquely determined.
We denote this configuration $C'$ by $\nextc_{\calA}(C)$.
The execution $\Xi_\calA(C_0)$ of algorithm $\calA$
starting from a configuration $C_0$ is defined
as an infinite sequence $C_0,C_1,\dots$ of configurations
such that $C_{t+1} = \nextc_{\calA}(C_t)$ for all $t=0,1,\dots$.
We say that a configuration $C_0$ is \emph{initial} if
the states of all agents are $\sinit$
and the incoming ports of all agents are $\bot$ in $C_0$.
Moreover, in the rooted setting,
we restrict the initial configurations to those where
all agents are located at a single node.

\begin{definition}[Dispersion Problem]
A configuration $C$ of an algorithm $\calA$
is called \emph{legitimate} if
(i) all agents in $\aset$ are located in different nodes in $C$,
and (ii) no agent changes its location
in execution $\Xi_\calA(C)$. 
We say that $\calA$ solves the dispersion problem
if execution $\Xi_\calA(C_0)$ reaches a legitimate configuration
for any initial configuration $C_0$.
\end{definition}

We evaluate the \emph{time complexity} or \emph{running time} of algorithm $\calA$
as the maximum number of steps until $\Xi_\calA(C_0)$
reaches a legitimate configuration, where the maximum is taken over all initial configurations $C_0$.
Let $\calM'_{\calA} \subseteq \calM_{\calA}$ be the set of all agent-states that can appear in any possible execution of $\calA$ starting from any initial configuration. 
We evaluate the \emph{space complexity} or \emph{memory space} of algorithm $\calA$
as $\log_2 |\calM'_{\calA}| + \log_2 \idmax$, \ie
the maximum number of bits required to represent an agent-state that may appear in those executions,
plus the number of bits required for each agent to store its own identifier.
This implies that we exclude the size of the working memory used for deciding the destination and updating states, as well as the space for storing input information, except for the agent's own identifier.


Throughout this paper, we denote by $[i,j]$ the set of integers
$\{i,i+1,...,j\}$. We have $[i,j]=\emptyset$ when $j<i$.
When the base of a logarithm is not specified, it is assumed to be 2.
We frequently use $\tau = \min(k,\Delta)$.
We define $\lctn(a,t)$ as the node where agent $a$ resides at time step $t$.
We also omit time step $t$ from any function in the form $f(*,t)$
and just write $f(*)$ if $t$ is clear from the context.
For example, we just write $\aset(v)$ and $\nu(a)$ instead of $\aset(v,t)$ and $\nu(a,t)$.

We have the following remark considering the fact that $G$ can be a simple path.
\begin{remark}
For any dispersion algorithm $\calA$, 
there exists a graph $G$ such that
an execution of $\calA$ requires $\Omega(k)$ time steps to achieve dispersion
on both the rooted and the general settings. 
\end{remark}

In the two algorithms we present in this paper, $\rootal$ and $\general$,
each agent maintains a variable $a.\settled \in \{\bot,\top\}$.
We say that an agent $a$ is a \emph{settler} when $a.\settled = \top$, and an \emph{explorer} otherwise.
All agents are explorers initially. Once an explorer becomes a settler, it never becomes an explorer again.
Let $t$ be the time at which an agent $a$ becomes a settler.
Thereafter, we call the location of $a$ at that time, \ie $\lctn(a,t)$, the \emph{home} of $a$.
Formally, $a$'s home at $t' \ge 0$, denoted by $\xi(a,t')$, is defined as
$\xi(a,t')=\bot$ if $t'<t$ and $\xi(a,t')=\lctn(a,t)$ otherwise.
It is worth mentioning that a settler may temporarily leave its home. Hence $\xi(a,t')=\lctn(a,t')$ may not always hold even after $a$ becomes a settler, \ie even if $t'\ge t$.
However, by definition, no agent changes its home.
We say that an agent $a$ \emph{settles} when it becomes a settler.

When a node $u$ is a home of an agent at time step $t$,
we call this agent the settler of $u$ and denote it as $\psi(u,t)$.
Formally, if there exists an agent $a$ such that $\xi(a,t)=u$, then $\psi(u,t)=a$; otherwise, $\psi(u,t)=\bot$. 
This function $\psi$ is well defined for the two presentented algorithms because
they ensure that no two agents share a common home.
We say that a node $u$ is \emph{unsettled} at time step $t$
if $\psi(u,t)=\bot$, and \emph{settled} otherwise.

\section{Rooted Dispersion}
\label{sec:rooted}
In this section, we present an algorithm, $\rootal$, that solves the dispersion problem in the rooted setting. That is, it operates under the assumption that all agents are initially located at a single node $s\in V$. This algorithm straightforwardly implements the strategy of the HEO-DFS, which we presented in Section \ref{sec:intro}. The time and space complexities of this algorithm are $O(k \log \tau)$ steps and $O(\log (k+\Delta))$ bits, respectively.

In an execution of Algorithm $\rootal$, the agent with the largest ID, denoted as $\amax$, serves as the leader.
Note that every agent can easily determine whether it is $\amax$ or not at time step $0$ by comparing the IDs of all agents.
Then, $\amax$ conducts a depth-first search (DFS), while the other agents move with the leader and one of them settles at an unsettled node when they visit it.
If $\amax$ encounters an unsettled node without any accompanying agents, $\amax$ settles itself on that node, achieving dispersion. During a DFS, $\amax$ must determine (i) whether there is an unsettled neighbor of the current location, and (ii) if so, which neighbor is unsettled. To make this decision, all DFS-based algorithms in the literature have $\amax$ visit neighbors one by one until it finds an unsettled node, which clearly requires $\Omega(\tau)$ steps. $\rootal$, in contrast, makes this decision in $O(\log \tau)$ steps with the help of the agents that have already settled on the neighbors of the current location. 


\begin{figure}[t]
\centering
\includegraphics[width=0.85\linewidth,clip]{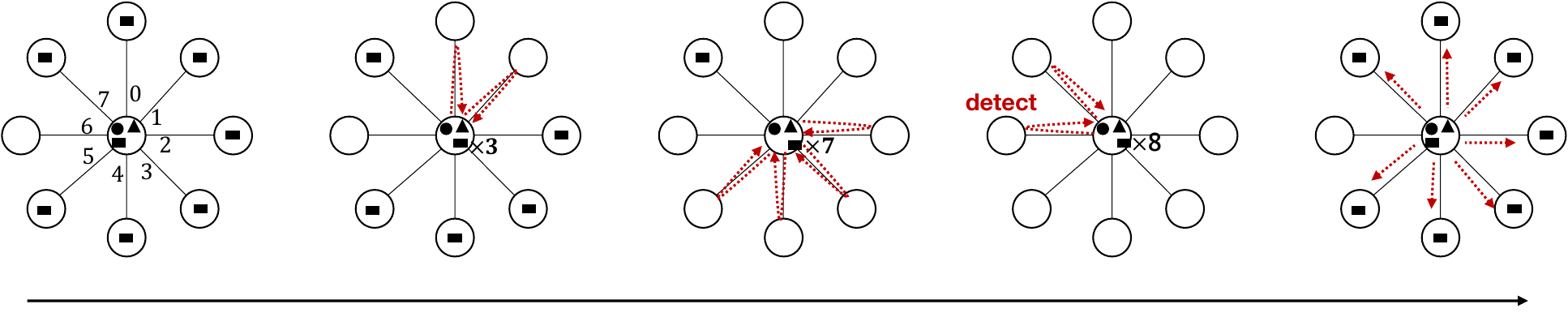}
\caption{
The behavior of the agents when the leader $\amax$ invokes $\Probe()$ at the center node $w$
in $\rootal$.
A black circle, triangle, and rectangle
represent a leader ($\amax$), a non-leader explorer, 
and a settler, respectively.
The integers in the leftmost figure represents port numbers.
In every two time steps,
the number of agents on $w$ excluding $\psi(w)$ doubles 
(\ie $2 \to 4 \to 8$)
until some agent detects an unsettled neighbor of $w$. 
After that, $\amax$ lets the helping settlers go back to their homes.
}
\label{fig:probing}
\end{figure}

\begin{algorithm}[t]
\caption{$\rootal$}
\label{al:rootal}
$b.\settled \gets \top$, where $b$ is the agent with the smallest ID in $\aset(s)$\\
\tcp*{$b$ settles at the starting node $s$}
$b.\parent \gets \bot$\;
\While{$\amax.\settled = \bot$}{
$\Probe()$\;
Let $w = \lctn(\amax)$\;
\uIf{$\psi(w).\nxt \neq \bot$}{
Let $u=N(w,\psi(w).\nxt)$ \tcp*{$u$ is an unsetteled node here}
All explorers in $\aset(w)$ go to $u$\;
$b'.\settled \gets \top$, where $b'$ is the agent with the smallest ID in $\aset(u)$\;
$b'.\parent \gets \amax.\pin$
}
\Else{
All explorers in $\aset(w)$ go back to node $N(w,\psi(w).\parent)$.\;
}
}
\modulespace

\Fn{$\Probe()$}{
Let $w=\lctn(\amax)$.\;
$(\psi(w).\nxt, \psi(w).\checked) \gets (\bot,-1)$\;
\While{$\psi(w).\checked \neq \delta_w-1$}{
Let $a_1,a_2,\dots,a_x$ be the agents in $A(w) \setminus \{\psi(w)\}$,
and let $\Delta'=\min(x,\delta_w-1-\psi(w).\checked)$.
For each $i=1,2,\dots,\Delta'$, assign $a_i$ to the neighboring node $u_i=N(w,i+\psi(w).\checked)$, and let $a_i$ make a round trip between $w$ and $u_i$. In other words, make $a_i$ move in the order $w \to u_i \to w$. If $a_i$ finds a settler at $u_i$, it will bring the settler $\psi(u_i)$ back to $w$.\;
\uIf{there exists $a_i$ that did not bring $\psi(u_i)$ back to $w$}{
$\psi(w).\nxt \gets i+\psi(w).\checked$ \tcp*{$u_i$ must be unsettled}
Break the while loop.
}\Else{
 $\psi(w).\checked \gets \psi(w).\checked + \Delta'$\;
}
}
Let all settlers except for $\psi(w)$ go back to their homes. 
}
\end{algorithm}

The pseudocode for Algorithm $\rootal$ is shown in Algorithm \ref{al:rootal}. This pseudocode consists of two parts: the main function (lines 1--12) and the function $\Probe()$ (lines 13--23).
As mentioned in the previous section, every agent $a$
maintains a variable $a.\settled \in \{\bot, \top\}$,
which decides whether $a$ is an explorer or a settler. 
In addition,
the settler $\psi(w)$ of a node $w$ maintains two variables, $\psi(w).\parent, \psi(w).\nxt \in [0, \delta_w - 1] \cup {\bot}$ for the main function.
As we will see later, the following are guaranteed each time $\amax$ invokes $\Probe()$ at node $w$:
\begin{itemize}
\item If there exists an unsettled node in $N(w)$,
the corresponding port number will be stored in $\psi(w).\nxt$.
More precisely, an integer $i$ such that $N(w, i) = u$ and $\psi(u) = \bot$
is assigned to $\psi(w).\nxt$.
\item 
If all neighbors are settled, 
$\psi(w).\nxt$ will be set to $\bot$.
\item $\Probe()$ will return in $O(\log \tau)$ time.
\end{itemize}

The main function performs a depth-first search using function $\Probe()$ to achieve dispersion. At the beginning of the execution, all agents are located at the same node $s$. Initially, the agent with the smallest ID settles at node $s$, and $\psi(s).\parent$ is set to $\bot$ (lines 1--2). 
Then, as long as there are unsettled nodes in $N(\lctn(\amax))$,
all explorers move to one of those nodes together (lines 7--8). We call this kind of movements \emph{forward moves}.
After each forward move from a node $w$ to $u$,
the agent with the smallest ID among $A(u)$ settles on $u$,
and $\psi(u).\parent$ is set to $i$ with $N(u, i)=w$ (lines 9--10).
 For any node $u \in V$, if $\psi(u).\parent \neq \bot$,
 we say that $w=N(u,\psi(u).\parent)$ is a \emph{parent} of $u$. 
 By line 9--10,
 each of the nodes except for the starting node $s$ will have its parent as soon as it becomes settled.
 When the current location has no unsettled neighbors, all explorers move to the \emph{parent} of the current location (lines 11--12). 
 We call this kind of movements \emph{backward moves} or \emph{retreats}. 
 Finally, $\amax$ terminates when it settles (line 3).

Since the number of agents is $k$, the DFS-traversal stops after $\amax$ makes a forward move $k-1$ times.
The agent $\amax$ makes a backward move at most once from any node. 
Therefore, excluding the execution time of $\Probe()$, the execution of the main function completes in $O(k)$ time. Furthermore, the function $\Probe()$ is invoked at most $2(k-1)$ times, once after each forward move and once after each backward move. Since a single invocation of $\Probe()$ requires $O(\log \tau)$ time,
the overall execution time of $\rootal$ can be bounded by $O(k \log \tau)$ time.

Let us describe the behavior of the function $\Probe()$, assuming that it is invoked on node $w$ at time step $t$.
Figure \ref{fig:probing} may help the readers to understand the behavior. 
In the execution of $\Probe()$, the leader $\amax$ employs the explorers present on $w$ and (a portion of) the settlers at $N(w)$ to search for an unsettled node in $N(w)$.
We implement this process with a variable $\psi(w).\checked \in [-1,\delta_w-1]$ for the settler $\psi(w)$.
Specifically, explorers at node $w$ verify whether the neighbors of $w$ are unsettled or not in the order of port numbers and store the most recently checked port number in $\psi(w).\checked$. Consequently, $\psi(w).\checked=\ell$ implies that the neighbors $N(w,0),N(w,1),\dots,N(w,\ell)$ are settled.
Let $x=|\aset(w,t)\setminus\{\psi(w)\}|$, i.e., there are $x$ agents at $w$ when $\Probe()$ is invoked, excluding $\psi(w)$. In the first iteration of the while loop (lines 16--22), the $\min(x,\delta_w)$ agents concurrently visit $\min(x,\delta_w)$ neighbors and then return to $w$ (lines 17--18). This entire process takes exactly two time steps. These agents bring back all the settlers, at most one for each neighbor, they find. If there is an agent that does not find a settler, then the node visited by that agent must be unsettled. In such a case, the port used by one of these agents is stored in $\psi(w).\nxt$, and the while loop terminates (lines 20--21).\footnote{
For simplicity, we reset $\psi(w).\checked$ to $-1$ each time we invoke $\Probe()$ at $w$,
so we do not use the information about which ports were already checked in the past invocation of $\Probe()$.
As a result, the value of $\psi(w).\nxt$ computed by $\Probe$
does not have to be the minimum port leading to an unsettled neighbor of $w$.
}
If all $x$ agents bring back one agent each, then there are $2x$ agents on $w$, excluding $\psi(w)$.
In the second iteration of the while loop, these $2x$ agents visit the next $2x$ neighbors and search for unsettled neighbors in a similar way. As long as no unsettled neighbors are discovered, the number of agents on $w$, excluding $\psi(w)$, doubles with each iteration of the while loop. Since there are at most $\tau=\min(k,\Delta)$ settled nodes in $N(w)$, after running the while loop at most $O(\log (\tau/x))=O(\log \tau)$ times, either an unsettled node will be found, or the search will be concluded without finding any unsettled nodes. In the latter case, since $\psi(w).\nxt$ is initialized to $\bot$ when $\Probe()$ is called (line 15), $\psi(w).\nxt=\bot$ will also be valid at the end of the while loop, allowing $\amax$ to verify that all neighbors of $w$ are settled. After the while loop ends, the settlers brought back to $w$ return to their homes (line 23).
This process of ``returning to their homes'' requires the agents to remember the port number leading to their home from $w$. However, we exclude this process from the pseudocode because it can be implemented in a straightforward manner, and it requires only $O(\log \Delta)$ bits of each agent's memory.
In conclusion, we have the following lemma.

\begin{lemma}
\label{lemma:probe}
Each time $\Probe()$ is invoked on node $w\in V$, $\Probe()$ finishes in $O(\log \tau)$ time. At the end of $\Probe()$, it is guaranteed that:
(i) if there exists an unsettled node in $N(w)$, then $N(w,\psi(w).\nxt)$ is unsettled, and 
(ii) if there are no unsettled nodes in $N(w)$, then $\psi(w).\nxt=\bot$ holds true.
\end{lemma}

\begin{lemma}
\label{lemma:space_rootal}
Each agent requires $O(\log (k+\Delta))$ bits of memory to execute $\rootal$.
\end{lemma}
\begin{proof}
In this algorithm, an agent handles several $O(\log \Delta)$-bit variable, $\nxt$, $\checked$, $\parent$, as well as the port number that the settler $\psi(u)$ needs to remember in order to return to node $u$ from node $w$ at line 23 after coming at line 17. Every other variable can be stored in a constant space. Therefore, the space complexity is $O(\log (k+\Delta))$ bits, adding the memory space to store the agent's identifier.
\end{proof}

\begin{theorem}
\label{theorem:rootal}
In the rooted setting,
algorithm $\rootal$ 
solves the dispersion problem within $O(k\log \tau)$ time using $O(\log (k+\Delta))$ bits of space per agent.
\end{theorem}

\begin{proof}
As long as there is an unsettled neighbor of the current location,
$\amax$ makes a forward move to one of those nodes. If there is no such neighbor, 
$\amax$ makes a backward move to the parent node of the current location. 
Since the graph is connected, this DFS-traversal clearly visits $k$ nodes with exactly $k-1$ forward moves
and at most $k-1$ backward moves. 
Thus, the number of calls to $\Probe()$ is at most $2(k-1)$ times. By Lemma \ref{lemma:probe}, the execution of $\rootal$ achieves dispersion within $O(k \log \tau)$ time.
\end{proof}

\section{General Dispersion}
\label{sec:general}
 \subsection{Overview}
 \label{sec:overview}

In this section, we present an algorithm $\general$ that solves the dispersion problem in $O(k \log \tau \cdot \log k) = O(k\log^2 k)$ time, using $O(\log (k+\Delta))$ bits of each agent's memory, in the general setting.
 Unlike the rooted setting, the agents are deployed arbitrarily. 
 In $\general$, we view the agents located at the same starting node as a single \emph{group} and achieve rapid dispersion by having each group perform a HEO-DFS in parallel, sometimes merging groups. 
 We show that by employing the group merge method given by Shintaku et al.~\cite{SSKM20}, say \emph{Zombie Method}, we can parallelize HEO-DFS by accepting an additive factor of $\log k$ to the space complexity and a multiplicative factor of $\log k$ to the time complexity. We have made substantial modifications to the Zombie Method to avoid conflicts between the function $\Probe()$ of HEO-DFS and the behavior of the Zombie Method.


As defined in Section \ref{sec:model}, 
agents $a$ with $a.\settled=\top$ are called settlers,
and the other agents are called explorers. 
In addition, in $\general$, 
we classify explorers to two classes, \emph{leaders} and \emph{zombies},
depending on a variable $\leader \in [0,\idmax]$.
We call an explorer $a$ a leader if $a.\leader = a.\id$, otherwise a zombie. 
Each agent $a$ initially has $a.\leader = a.\id$, so all agents are leaders at the start of an execution of $\general$. 
As we will see later, a leader may become a zombie and a zombie will eventually become a settler, whereas
a zombie never becomes a leader again, and a settler never becomes a leader or zombie again.
Among the agents in $A(v,t)$, the set of leaders (resp., zombies, settlers) staying at $v$ in time step $t$ is denoted by $A_L(v,t)$ (resp., $A_Z(v,t), A_S(v,t)$). 
By definition, $A(v,t)=A_L(v,t) \cup A_Z(v,t) \cup A_S(v,t)$.

We introduce a variable $\level \in \mathbb{N}$ to bound the execution time of $\general$.
We call the value of $a.\level$ the \emph{level} of agent $a$. 
The level of every agent is $1$ initially. 
The pair $(a.\leader, a.\level)$ serves as the group identifier:
when agent $a$ is a leader or settler, 
we say that $a$ belongs to a group $(a.\leader, a.\level)$.
By definition, for any $(\ell,i) \in \mathbb{N}^2$,
a group $(\ell,i)$ has at most one leader. 
A zombie does not belong to any group.
However, when it accompanies a leader, it joins the HEO-DFS of that leader. 
We define a relationship $\prec$ between any two non-zombies $a$ and $b$ using these group identifiers as follows:
\vspace{6pt}

\noindent
$$a \prec b \ \ \ \Leftrightarrow\ \ \  (a.\level < b.\level) \lor (a.\level=b.\level \land a.\leader < b.\leader).$$
We say that agent $a$ is \emph{weaker} than $b$ if $a \prec b$, and that $a$ is \emph{stronger} than $b$ otherwise.

Initially, all agents are leaders and each forms a group of size one. In the first time step, the strongest agent at each node turns all the other co-located agents into zombies (if exists).
From then on, each leader performs a HEO-DFS while leading those zombies.
For any leader $a$, we define the \emph{territory} of $a$ as 
\vspace{6pt}

\noindent
$$V_a = \{v \in V \mid \exists b \in A: \psi(v) = b \land b.\leader = a.\id \land b.\level = a.\level \}.$$
Each time a leader $a$ visits an unsettled node, it settles one of the accompanying zombies (if exists),
giving it $a$'s group identifier $(a.\leader, a.\level)$. That is, $a$ expands its territory.
If a node outside $a$'s territory is detected during the probing process of HEO-DFS, that node is considered unsettled even though it belongs to the territory of another leader. As a result, $a$ may move forward to a node $u$ that is inside another leader's territory.
If that node $u$ belongs to the territory of a weaker group,
$a$ incorporates the settler $\psi(u)$ into its own group by giving $\psi(u)$ its group identifier $(a.\leader, a.\level)$. If a leader $a$ encounters a stronger leader or a stronger settler during its HEO-DFS, $a$ becomes a zombie and terminates its own HEO-DFS.  If there is a leader at the current location $\lctn(a)$ when $a$ becomes a zombie, $a$ joins the HEO-DFS of that leader. Otherwise, the agent $a$, now a zombie, chases a stronger leader by moving through the port $\psi(v).\nxt$ at each node $v$.
Unlike $\rootal$, a leader updates $\psi(v).\nxt$ with the most recently used port
even when it makes a backward move. This ensures that $a$ catches up to a leader eventually,
at which point $a$ joins the HEO-DFS led by the leader.


Unlike $\rootal$, a leader does not settle itself at a node in the final stage of HEO-DFS. The leader $a$ suspends the HEO-DFS if it visits an unsettled node but it has no accompanying zombies to settle at that time.
A leader who has suspended the HEO-DFS due to the absence of accompanying zombies is called a \emph{waiting leader}.
Conversely, a leader with accompanying zombies is called an \emph{active leader}.
A waiting leader $a$ resumes the HEO-DFS when a zombie catches up to $a$ at $\lctn(a)$.
As we will see later, the execution of $\general$ ensures that all agents eventually become either waiting leaders or settlers, each residing at a distinct node. 
The agents have solved the dispersion once such a configuration is reached
because thereafter no agent moves and no two agents are co-located.

When a leader $a$ encounters a zombie $z$ with the same level,
$a$ increments its level by one, and $z$ resets its level to zero.
This ``level up'' changes the identifier of $a$'s group, \ie
from $(a.\id,i)$ to $(a.\id, i+1)$ for some $i$.
By the definition of the territory, at this point, $a$ loses all nodes from its territory except for the current location.
That is, each time a leader $a$ increases its level, it restarts its HEO-DFS from the beginning.
Note that this ``level up'' event also occurs when two leaders $a,b\ (b \prec a)$ with the same level meet (and there is no stronger agent at the location) because then $b$ becomes a zombie after it finds a stronger leader $a$, which results in the event that a leader $a$ encounters a zombie with the same level, say $b$.
We have the following lemma here.
\begin{lemma}
\label{lemma:level}
The level of an agent is always at most $\log_2 k + 1$.
\end{lemma}
\begin{proof}
A level-up event requires one leader $a$ and one zombie $b$ with the same level. 
That zombie $b$ will get level 0.
Thereafter, $b$ never triggers a level-up event again
because the level of a leader is monotonically non-decreasing starting from level 1.
Therefore, for any $i \ge 1$, the number of agents that can reach level $i$
is at most $\lfloor k/2^{i-1}\rfloor$,
leading the lemma.
\end{proof}

Therefore, each leader performs HEO-DFS at most $O(\log k)$ times. According to the analysis in Section \ref{sec:rooted}, each HEO-DFS completes in $O(k\log \tau)$ time, which seems to imply that $\general$ finishes in $O((\log k)\cdot (k \log \tau))=O(k \log^2 k)$ time. However, this analysis
does not take into account the length of the period during which leaders suspend their HEO-DFS. 
Thus, it is not clear whether a naive implementation of the strategy described above would achieve the dispersion in $O(k \log^2 k)$ time. Following Shintaku et al.~\cite{SSKM20}, we vary the speed of zombies chasing leaders based on a certain condition, which bounds the execution time by $O(k \log^2 k)$ time.


We give zombies different chasing speeds as follows.
First, we classify zombies based on two variables $\levelL$ and $\levelS$ that each zombie manages.
For any zombie $z$, we call $z.\levelL$ and $z.\levelS$ the \emph{location level}
and \emph{swarm level} of $z$.
When a leader $z$ becomes a zombie, it initializes both $z.\levelL$ and $z.\levelS$
with its level, \ie $z.\level$. 
Thereafter, a zombie $z$ copies the level of $\psi(\lctn(z))$ to $z.\levelL$
and updates $z.\levelS$ to be $\max\{b.\level \mid b \in A_Z(\lctn(z))\}$
in every $O(1)$ time steps.
Since a zombie only chases a leader with an equal or greater level, 
$z.\levelS \le z.\levelL$ always holds.
We say that a zombie $z$ is \emph{strong} if $z.\levelS = z.\levelL$;
$z$ is \emph{weak} otherwise. 
Then, we exploit the assumption that the agents are synchronous and let weak zombies move twice as frequently as strong zombies to chase a leader. 
As we will prove later, this difference in chasing speed results in a desirable property of $\general$, namely that
$\min (\{a.\level \mid a \in A_{AL}\} \cup \{z.\levelL \mid z \in A_Z\})$
is monotone non-decreasing and increases by at least one in every $O(k \log \tau)$ steps,
where $A_{AL}$ is the set of active leaders and $A_Z$ is the set of zombies both in the whole graph,
until $A_{AL} \cup A_Z$ becomes empty. 
Thus, by Lemma \ref{lemma:level},
$A_{AL} \cup A_Z$ becomes empty and the dispersion is achieved in 
$O(k \log \tau \cdot \log k)= O(k \log^2 k)$ steps.

\begin{table}[t]
\centering
\ \caption{Slot Assignments}
\label{table:slots}
\vspace{0.3cm} 
\begin{tabular}{c c c c}
\hline
Slot Number & Role & Initiative &Pseudocode\\
\hline
Slot 1& 
Leader election
& Leaders& \ref{al:leader}\\
Slot 2& Settle, increment level, etc. &Leaders&\ref{al:probe}\\
Slots 3& Move to join $\Probe()$ &Settlers&\ref{al:settled}\\
Slots 4--8& $\Probe()$&Leaders&\ref{al:probe}\\
Slot 9--10& Chase for leaders&Zombies&\ref{al:zombie}\\
Slot 11--12& Move forward/backward&Leaders&\ref{al:leader}\\
\hline
\end{tabular}
\end{table}

\begin{figure}[t]
\centering
\includegraphics[width=0.65\linewidth,clip]{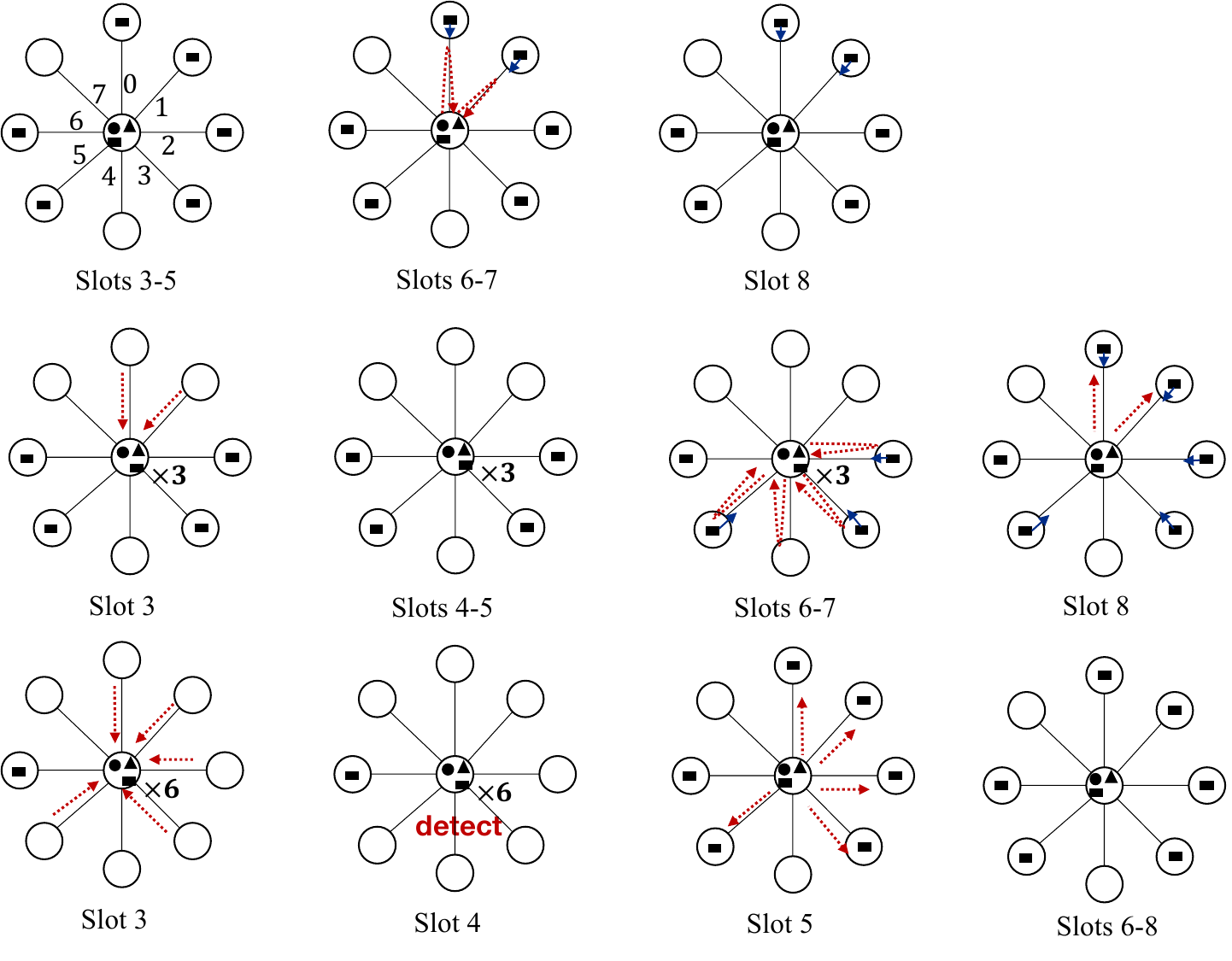}
\caption{
The behavior of explorers when their leader invokes $\Probe()$ at the center node $w$
in $\general$.
A black circle, triangle, and rectangle
represent a leader, a zombie, 
and a settler, respectively.
The integers in the top left figure represents port numbers.
}
\label{fig:probingGeneral}
\end{figure}

\subsection{Time Slots}
\label{sec:implement}
In $\general$, we group every 12 time steps into one unit, with each unit consisting of twelve \emph{slots}.
In other words, time steps $0,1,2,\dots$ are classified into twelve slots.
Specifically, each time step $t \ge 0$ is assigned to slot $(t \bmod 12) + 1$. For example, time step 26 is in slot 3, and time step 47 is in slot 12. Dividing all time steps into twelve slots helps to reduce the interference of multiple HEO-DFSs and allows us to set different ``chasing speeds'' for weak and strong zombies. Table \ref{table:slots} summarizes the roles of each slot. 

Essentially, slots 1--2 are designated for leader election (i.e., group merging), slots 3--8 for probing, slots 9--10 for zombie chasing, and slots 11--12 for forward and backward movement in DFS traversal. It is important to note that settlers always stay at their home during slots 1–2 and 9–12. Hence, once a settler leaves its home, it returns within $O(1)$ steps. This is not the case in $\rootal$, where a settler in helping mode does not return home until its leader completes the probing process. In $\general$, this frequent return home enables leaders to detect collisions with other groups: if a leader enters another group's territory, it will certainly notice the intrusion during the next slot 1, as it encounters a settler from that group.


Thus, the probing process in $\general$ slightly differs from that in $\rootal$. Consider a leader $a_l$ starting the probing process at a node $w$ (refer to Figure \ref{fig:probingGeneral}). The objective here is to identify any neighboring node of $w$ that lies outside $a_l$'s territory, if such exists. During slots 6--7, explorers at $w$ visit its neighbors and return in parallel. If an explorer $b$ encounters a settler $s$ at a node $u \in N(w)$ in slot 6, $b$ does not bring $s$ back to $w$ in slot 7. Instead, $b$ requests $s$ to enter helping mode, wherein $s$ records the port number to $w$ in the variable $s.\help \in \mathbb{N} \cup {\bot}$ (A settler $s$ is in helping mode if and only if $s.\mode \neq \bot$). The helping settler $s$ moves to $w$ via port $s.\help$ in the subsequent slot 3, joins the probing in slots 6-7, and returns to its home $u$ again in slot 8. Leader $a_l$ expects that the exact $\psi(w).\checked$ helping settlers arrives at $w$ in each slot 3. If this does not occur, $a_l$ detects a non-territorial neighbor in slot 4 by checking the $\pin$ variable of the helping settlers at $w$. Similar to $\rootal$'s probing process, the total number of explorers and helping settlers at $w$ doubles until such a neighbor is detected, concluding the process in $O(\log \tau)$ steps. Subsequently, $a_l$ reverts the helping settlers at $w$ to non-helping mode by setting their $\help$ variable to $\bot$ and instructs them to return to their homes in slot 5.


As mentioned earlier, we differentiate the chasing speed of weak zombies and strong zombies. Specifically, weak zombies move in both slots 9 and 10, while strong zombies move only in slot 10.

Before presenting the detailed implementation of $\general$ including pseudocodes in Section \ref{sec:detail},
we give a proof sketch here for the following main theorem.

\begin{theorem}
\label{theorem:general}
In the general setting, there exists an algorithm that solves the dispersion problem within $O(k \log \tau \cdot \log k)$ time using $O(\log (k+\Delta))$ bits of space per agent.
\end{theorem}
\begin{psketch}

It suffices to show that $A_{AL} \cup A_Z$ becomes empty within $O(k \log \tau \cdot \log k)$ steps, at which point every agent is either a waiting leader or a settler, thereby achieving dispersion.
We obtain this bound from Lemma \ref{lemma:level} and the fact that
$\alpha = \min (\{a.\level \mid a \in A_{AL}\} \cup \{z.\levelL \mid z \in A_Z\})$
increases by at least one in every $O(k \log \tau)$ time steps
unless $A_{AL} \cup A_Z$ becomes empty (Lemma \ref{lemma:main} in Section \ref{sec:proof_general}).
We can prove Lemma \ref{lemma:main} roughly as follows.
Suppose $\alpha = i$. First,
all weak zombies at location level $i$ vanish within $O(k)$ steps as they move faster than leaders and strong zombies, eventually encountering a leader, higher-level settlers or stronger zombies. Thereafter, no new weak zombies at location level $i$ are created. Subsequently, without weak zombies at location level $i$, waiting leaders at level $i$ do not resume active HEO-DFS without increasing its level, which leads to the disappearance of active leaders at level $i$ within $O(k \log \tau)$ steps. Finally, strong zombies chasing leaders at level $i$ catch up to those leaders within $O(k)$ steps or find higher-level settlers, resulting in the increase of their location level. Hence, all zombies with location level $i$ and active leaders at level $i$ are eliminated within $O(k \log \tau)$ steps.
\qed
\end{psketch}

\setcounter{AlgoLine}{0}
\begin{algorithm}[t]
\caption{The behavior of a \emph{leader} $a$}
\label{al:leader}
\While{$\tr$}{
 \CO{Slot 1 begins}
 Let $w=\lctn(a)$.\;
  \If{$\exists b \in A_L(w) \cup A_S(w):a\prec b$ }{
    $a.\levelL \gets a.\levelS \gets a.\level$\;
    $a.\leader \gets b.\leader$
    \tcp*{$a$ becomes a zombie and stops Algorithm \ref{al:leader}}
 } 
 \CO{Slot 2 begins}
 \If(\tcp*[f]{$a$ is an active leader if $A(w)\neq\{a\}$}){$A(w)\neq\{a\}$}{ 
\If{$\psi(w)= \bot\  \lor \ \psi(w) \prec a$}{
\If{$\psi(w) = \bot$}{
Settle one zombie in $A_Z(w)$ at $w$\;
}
 $\psi(w).\parent \gets a.\parent$\tcp*{Initially, $a.\parent = \bot$}
}
 \If{$\exists b \in A_Z(w): a.\level = b.\level$}{
 $(a.\level,b.\level) \gets (a.\level + 1,0)$\;
 $\psi(w).\parent \gets \bot$\;
 $a.\needInit \gets \tr$\;
}
$(\psi(w).\leader,\psi(w).\level) \gets (a.\id,a.\level)$\;
\If(\tcp*[f]{Initially, $a.\needInit = \tr$}){$a.\needInit = \tr$}{
$(\psi(w).\nxt, \psi(w).\checked,\psi(w).\help,\psi(w).\done) \gets (\bot, -1,\bot,\fl)$\;
$a.\needInit \gets \fl$
}
$\Probe(a)$ \tcp*{See Algorithm \ref{al:probe}}
\If{$\psi(w).\done = \tr$}{
\CO{Slot 11 begins}
\If{$\psi(w).\nxt = \bot$}{$\psi(w).\nxt \gets \psi(w).\parent$ \tcp*{for backward move}}
All agents in $A(w)\setminus \{\psi(w)\}$ move to $N(w,\psi(w).\nxt)$\;
\CO{Slot 12 begins}
$a.\parent \gets a.\pin$\;
$a.\needInit \gets \tr$\;
}
}
}
\end{algorithm}

\begin{algorithm}
\caption{$\Probe(a)$}
\label{al:probe}
\CO{Slot 4 begins}
Let $w=\lctn(a)$.\;
$\psi(w).\nxt(w) \gets \begin{cases}\min P & \text{if\ } P \neq \emptyset\\ \bot & \text{otherwise},\end{cases}$\\
\ \ \ where $P=[0,\psi(w).\checked] \setminus \{b.\pin \mid b \in A_S(w)\setminus\{\psi(a)\}\}$ \;
$b.\help \gets \bot$ for all $b \in A_S(w)$ with $b \prec a$.\;
Let all agents $b \in A_S(w)$ with $b \prec a$ go back to their homes\;
\uIf{$\psi(w).\nxt \neq \bot \lor \psi(w).\checked = \delta_w - 1$}{
\CO{Slot 5 begins}
Execute $b.\help \gets \bot$ for each $b\in A_S(w)\setminus \{\psi(w)\}$\;
Let all agents in $A_S(w)\setminus \{\psi(w)\}$ go back to their homes.\;
$\psi(w).\done \gets \tr$\;
}\Else{
\CO{Slot 6 begins}
Let $\{a_1,a_2,\dots,a_x\}$ be the set of agents in $A(w) \setminus \{\psi(w)\}$\;
Let $\Delta'=\min(x,\delta_w-1-\psi(w).\checked)$\;
Let $u_i=N(w,i+\psi(w).\checked)$ for $i=1,2,\dots,\Delta'$\;
\For{\textbf{each} $a_i \in \{a_1,a_2,\dots,a_{\Delta'}\}$ \textbf{in parallel}}{
$a_i$ moves to $u_i$.\;
\CO{Slot 7 begins}
\uIf{$(a_i.\leader,a_i.\level)=(\psi(u_i).\leader,\psi(u_i).\level)$}{
$a_i.\found \gets \tr$\;
$\psi(u_i).\help \gets a_i.\pin$
}
\Else{
$a_i.\found \gets \fl$
}
Move to $N(u_i,a_i.\pin)$\;
}
\CO{Slot 8 begins}
\If{$\exists i \in [1,\Delta']: a_i.\found = \fl$}{
    $\psi(w).\nxt \gets i+\psi(w).\checked$
}
$\psi(w).\checked \gets \psi(w).\checked + \Delta'$\;
Let all agents in $A_S(w)\setminus \{\psi(w)\}$ go back to their homes.\;
}
\end{algorithm}

\begin{algorithm}
\caption{The behavior of a \emph{settler} $s$ in Slot 3}
\label{al:settled}
\CO{Slot 3 begins}
Move to $N(\lctn(s),s.\help)$ if $s.\help \neq \bot$.
\end{algorithm}

\begin{algorithm}
\caption{The behavior of a \emph{zombie} $z$ in Slots 9 and 10}
\label{al:zombie}
\CO{Slot 9 begins}
$(z.\levelL,z.\levelS) \gets (\psi(w).\level,\max\{z'.\level \mid z' \in A_Z(\lctn(z)\})$\;
\If{$A_L(\lctn(z))=\emptyset$ and $z$ is a weak zombie}{
Move to $N(\lctn(z),\psi(\lctn(z)).\nxt)$
}
\CO{Slot 10 begins}
\If{$A_L(\lctn(z))=\emptyset$}{
Move to $N(\lctn(z),\psi(\lctn(z)).\nxt)$
}
\end{algorithm}

\subsection{Detail Implementation of General Dispersion}
\label{sec:detail}

The pseudocode for the $\general$ algorithm is shown in Algorithms \ref{al:leader}, \ref{al:probe}, \ref{al:settled}, and \ref{al:zombie}. In slots 1, 2, 4--8, 11, and 12, agents operate only under the instruction of a leader. Algorithms \ref{al:leader} and \ref{al:probe} define how each leader $a$ operates and gives instructions in those slots.  
Algorithm \ref{al:settled} defines the behavior of settlers in slot 3.
Algorithm \ref{al:zombie} specifies the behavior of zombies in slots 9 and 10.
Note that each agent needs to manage an $O(1)$-bit variable to identify the slot of the current time step, but for simplicity, the process related to its update is not included in the pseudocode because it can be implemented in a naive way.


First, we explain the behavior of a leader $a$.
Let $w$ be the node where $a$ is located in slot 1.
We make leader election in slot 1 (lines 4--6).
Leader $a$ becomes a zombie when it finds a stronger leader or settler on $w$.
If $a$ becomes a zombie, it no longer runs Algorithms \ref{al:leader} and \ref{al:probe}, and runs only Algorithm \ref{al:zombie}. 
Consider that $a$ survives the leader election in slot 1. 
In slot 2, if there are no agents other than $a$ on $w$, $a$ is a waiting leader and does nothing until the next slot 1. Otherwise, the leader $a$ (i) settles one of the accompanying zombies if $w$ is unsettled, (ii) updates its level if it finds a zombie with the same level, and (iii) gives the settler $\psi(w)$ its group identifier $(a.\id,a.\level)$ (lines 9--17).
Note that settlers may leave their homes only in slots 3--8 (to join $\Probe()$),
thus $a$ can correctly determine whether $\psi(w) = \bot$ or not here (lines 9--10).
If $\psi(w) \prec a$, this procedure incorporates $\psi(w)$ into $a$'s group, \ie expands the territory of $a$.
Each leader $a$ manages a flag variable $a.\needInit \in \{\fl,\tr\}$, initially set to $\tr$.
This flag is raised each time $a$ requires probing, \ie after it makes a forward or backward move (line 29), 
and when it increases its level (line 16).  
If the flag is raised, it initializes the variables used for $\Probe()$, say $\psi(w).\nxt$, $\psi(w).\checked$, 
$\psi(w).\help$, and  $\psi(w).\done$ in slot 2 (line 19).

Thereafter, $a$ invokes $\Probe()$ at the end of slot 2.
This subroutine runs in slots 4--8. 
While $\Probe()$ in $\rootal$ returns the control to the main function after completing the probing,
\ie determining whether or not an unsettled neighbor exists, 
$\Probe()$ in $\general$ returns the control each time slot 8 ends
even if it does not complete the probing. 
Consider that there are $x-1$ accompanying zombies when a leader $a$ begins the probing. 
First, a leader $a$ and the $x-1$ accompanying zombies join the probing. Each of them, say $b$, moves from a node $w$ to one of its neighbors $u \in N(w)$ in slot 6 (line 47) and goes back to $w$ in slot 7 (line 54).
If $b$ finds a settler in the same group at $u$, 
it sets $\psi(u).\help$ to $b.\pin$ (line 51).
As long as $s.\help \neq \bot$, a settler $s$ at a node $v$ goes to a neighbor $N(v,s.\help)$ in slot 3 (line 60, Algorithm \ref{al:settled}). 
Hence, in the next slot 3, that settler $\psi(u)$ goes to $w$.
If there are $2x$ agents at $w$ excluding $\psi(w)$, 
those $2x$ agents perform the same process in the next slots 6 and 7,
that is, they go to unprobed neighbors, update the $\help$ of settlers in the same group (if exists),
and go back to $w$. In slot 8, $a$ sends the helping settlers back to their home.
The number of agents joining the probing at $w$,
\ie $|A(w) \setminus \{\psi(w)\}|$, doubles at each iteration of this process
until they find a node without a settler in the same group
or finish probing all neighbors in $N(w)$. 
Thus, like $\rootal$, the probing finishes in $O(\log \tau)$ time steps.
At this time, $\psi(w).\nxt = \bot$ holds if
all neighbors in $N(w)$ are settled by settlers in the same group.
Otherwise, $N(w,\psi(w).\nxt)$ is unsettled or settled by a settler in another group. 
Then, in the next slot 5, $a$ resets the $\help$ of all settlers at $w$ to $\bot$
except for $\psi(w)$, lets them go back to their homes, and sets $\psi(w).\done$ to $\tr$,
indicating that the probing is done (lines 38--40).
The probing process described above may be prevented by a stronger leader $b$
when $b$ visits a node $v \in N(w)$ such that $\psi(v)$ belongs to $a$'s group and $\psi(v).\help \neq \bot$.
Then, $b$ incorporates $\psi(v)$ into $b$'s group,
and set $\psi(v).\help$ to $\bot$ (line 19), so $\psi(v)$ never goes to $w$ to help $a$'s probing. 
However, this event actually speeds up $a$'s probing: $a$ identifies this event when noticing that $\psi(v)$ does not arrive at $w$ in the next slot 4. As a result, $a$ can set $\psi(w).\nxt$ to $p$ where $N(w,p) = v$ (lines 32--34).

Note that, even during the probing process at node $w$, leader $a$ might become a zombie if it meets a stronger leader $b$ in slot 1. Some settlers might then move to $w$ in the next slot 3 to help $a$, not knowing $a$ is now a zombie. In these situations, $b$ changes the $\help$ of these settlers to $\bot$ and sends them back to their homes in slot 4. Thereafter, those settlers remain at their home at least until they are incorporated into another group. 


If a leader $a$ at $w$ observes $\psi(w).\done = \tr$, 
it makes a forward or backward move in slot 11 (lines 22--29). 
Each time $a$ makes a forward or backward move to a node $u$,
it remembers $a.\pin$ in $a.\parent$ after the move (line 28).
This port number will be stored on the variable $\psi(u).\parent$ when $a$ settles a zombie on $u$ or $a$ incorporates $\psi(u)$ from the territory of another group.
Note that this event occurs only when the last move is forward.
Thus, like $\rootal$, $a$ constructs a DFS tree in its territory.
It is inevitable to use a variable $a.\parent$ tentatively since
$a.\pin$ is updated every step by definition of a special variable $\pin$
and $a$ may become a waiting leader after moving to $u$.
Unlike $\rootal$, $a$ records the most recently used port to move in $\psi(\lctn(a)).\nxt$ even when it makes a  backward move (lines 24--25). This allows a zombie to chase a leader. 


The behavior of zombies in slots 9 and 10 is very simple (lines 61--67).
A zombie always updates its location and swarm levels in slot 9 (line 62).  
A zombie $z$ not accompanying a leader always chases a leader by moving through the port $\psi(\lctn(z)).\nxt$.
As mentioned earlier, we differentiate the chasing speed of weak zombies and strong zombies. Specifically, weak zombies move in both slots 9 and 10, while strong zombies move only in slot 10 (lines 63--67).

\begin{lemma}
\label{lemma:levelL_is_increasing}
The location level of a zombie 
is monotonically non-decreasing.
\end{lemma}
\begin{proof}
Neither a leader nor a settler decreases its level in $\general$. 
When a zombie $z$ does not accompany a leader, it chases a leader through port $\psi(\lctn(z)).\nxt$.
This port $\psi(\lctn(z)).\nxt$ is updated only if a leader makes a forward or backward move
from $\lctn(z)$, and the leader updates the level of $\psi(N(\lctn(z),\psi(\lctn(z)).\nxt))$
if it is smaller than its level. Thus, a zombie never decreases its location level
by chasing a leader. 
When a zombie $z$ accompanies a leader, 
the leader copies its level to $\psi(\lctn(z)).\level$ in slot 2, which is copied to $z.\levelL$ in slot 8.
The leader that $z$ accompanies may change but does not change to a weaker leader. 
Thus, a zombie never decreases its location level when accompanying a leader.
\end{proof}

\subsection{Proofs of Theorem \ref{theorem:general}}
\label{sec:proof_general}
Remember that $A_Z$ and $A_{AL}$ are the set of zombies and the set of active leaders, respectively,
in the whole graph. We have the following lemma.
\begin{lemma}
\label{lemma:weak_disppear}
For any $i \ge 0$, the number of weak zombies with a location level $i$ is monotonically non-increasing
starting from any configuration where 
$\min (\{a.\level \mid a \in A_{AL}\} \cup \{z.\levelL \mid z \in A_Z\})=i$.
\end{lemma}
\begin{proof}
Let $C$ be a configuration
where $\min (\{a.\level \mid a \in A_{AL}\} \cup \{z.\levelL \mid z \in A_Z\})=i$.
When a leader with level $i$ becomes a zombie, its location level is $i$ (line 5).
So, a leader with level $i$ may become a strong zombie with a location level $i$
but never becomes a weak zombie with a location level $i$. 
The swarm level of a zombie decreases only when the zombie accompanies a leader (and this leader settles another zombie). 
Thus, a strong zombie with a location level $i$ that does not accompany a leader
cannot become a weak zombie without increasing its location level.
Moreover, starting from $C$, a strong zombie with location level $i$
must increase its location level when it encounters a leader in slot 1. 
Hence, the number of weak zombies with a location level $i$ is monotonically decreasing. 
\end{proof}

\begin{lemma}
\label{lemma:main}
$\min (\{a.\level \mid a \in A_{AL}\} \cup \{z.\levelL \mid z \in A_Z\})$
is monotone non-decreasing and increases by at least one in every $O(k \log \tau)$ time steps
unless $A_{AL} \cup A_Z$ becomes empty.
\end{lemma}

\begin{proof}
Let $i$ be an integer $i \ge 0$ and $C$ a configuration where $\min (\{a.\level \mid a \in A_{AL}\} \cup \{z.\levelL \mid z \in A_Z\})=i$.
It suffices to show that leaders with level $i$ and zombies with location level $i$ disappear
in $O(k \log \tau)$ time steps starting from $C$. 

Consider an execution starting from $C$.
By Lemma \ref{lemma:weak_disppear}, a weak zombie with location level $i$ is never newly created
in this execution. 
Let $z$ be any weak zombie with a location level $i$ that does not accompany a leader
in a configuration $C$.
In every 12 slots, $z$ moves twice, while a strong zombie and a leader move only once, excluding the movement for the probing. Therefore, $z$ catches up to a strong zombie and becomes strong too, 
catches up to a leader with level $i$, or increases its location level in $O(k)$ time steps. 
When $z$ catches up to a leader, it joins the HEO-DFS of the leader, or this leader becomes a zombie.
In the latter case, $z$ becomes a strong zombie. 
Thus, $z$ settles or becomes a strong zombie (with the current leader)
in $O(k \log \tau)$ time steps. Therefore, the number of weak zombies with location level $i$
becomes zero in $O(k \log \tau)$ steps.
After that, no waiting leader with level $i$
resumes its HEO-DFS without increasing its level because there is no weak zombie with location level $i$.
Therefore, every active leader with location level $i$ becomes a zombie with location level at least $i+1$
or a waiting leader in $O(k \log \tau)$ steps. 
Thus, active leaders with location level $i$ also disappear in $O(k \log \tau)$ steps.
From this time, no leader moves in the territory of a group with level $i$ or less. 
Hence, every strong zombie with location level $i$ increases its location level
or catches up to a waiting leader. Since the level of a waiting leader is at least $i$,
the latter event also increases $z$'s level by at least one.
\end{proof}

Since an agent in $\general$ manages only a constant number of variables, each with $O(\log (k+\Delta))$ bits,
Lemmas \ref{lemma:level} and \ref{lemma:main} 
yield Theorem \ref{theorem:general}.

\setcounter{AlgoLine}{0}
\begin{algorithm}
\caption{$\rootopt$}
\label{al:root_optimal}
$\Settle(b,s,\bot)$, where $s$ is the starting node
and $b$ is the agent with the smallest ID in $\aset(s)$\;
$\Probe()$\;
\While{$\amax.\settled = \bot$}{
Let $w = \lctn(\amax)$\;
\uIf{$\uemp(w) \neq \bot$}{
Chose any neighbor $u \in \uemp(w)$ \tcp*{$u$ is an unsetteled node here}
All explorers in $\aset(w)$ go to $u$\;
$\Settle(b,u,\amax.\pin)$,
where $b'$ is the agent with the smallest ID in $\aset(u)$\;
$\Probe()$
}
\Else{
All explorers in $\aset(w)$ go back to node $N(w,\psi(w).\parent)$.\;
}
}
\modulespace

\Fn{$\Settle(b,v,p)$}{
$b.\settled \gets \top$\;
$b.\parent \gets p$\;
\For{$i\in [0,\delta_v]$}{$b.\checked[i] \gets \bot$}
}
\modulespace

\Fn{$\Probe()$}{
Let $w=\lctn(\amax)$\;
Let $\ell=|A(w) \setminus \{\psi(w)\}|$ \tcp*{$\ell$ is the number of explorers}
\While{$\ubot(w) \neq \emptyset \land |\uemp(w)|< \ell$}{
Let $a_1,a_2,\dots,a_x$ be the agents in $A(w) \setminus \{\psi(w)\}$\tcp*{$x \le \ell$}
Let $\Delta' = \min(x,|\ubot(w)|)$\;
Choose $\Delta'$ neighbors $u_1,u_2,\dots,u_{\Delta'}$ from $\ubot$ arbitrarily.
\;
\For{\textbf{each} $a_i \in \{a_1,a_2,\dots,a_{\Delta'}\}$ \textbf{in parallel}}{
$a_i$ moves to $u_i$\;
\If{$a_i$ finds $\psi(u_i)$ at $u_i$}{
$\psi(u_i).\checked[a_i.\pin] \gets \full$\;
}
$a_i$ goes back to $w$ along with $\psi(u_i)$ (if it exists)\;
Let $p_i$ be a port number such that $N(w,p_i) = u_i$\;
$\psi(w).\checked[p_i] \gets \begin{cases}
\full & \text{if }a_i \text{ found }\psi(u_i) \text{ at }u_i\\
\emp& \text{otherwise},
    \end{cases}$\;
}
}
Let all settlers except for $\psi(w)$ go back to their homes\;
\If{$|\uemp(w)|\ge \ell$}{
All explorers in $A(w)$ go to distinct neighbors in $\uemp(w)$ and 
settle at those nodes.
}
}
\end{algorithm}

\section{For Further Improvement in Time Complexities}
\label{sec:rootopt}


In the previous sections, we introduced nearly time-optimal dispersion algorithms: an $O(k \log \tau)$-time algorithm for the rooted setting and an $O(k \log^2 k)$-time algorithm for the general setting. This raises a crucial question: is it possible to develop a truly time-optimal algorithm, specifically an $O(k)$-time algorithm, even if it requires much more space? In this section, we affirmatively answer this question for the rooted setting. We present an $O(k)$-time algorithm that utilizes $O(\Delta+\log k)$ bits
of space per agent. However, the feasibility of an $O(k)$-time algorithm in the general setting remains open.

We refer to the new algorithm as $\rootopt$ in this section. 
In $\rootal$, with $O(\log (k+\Delta))$ bits of space,
each settler $\psi(w)$ cannot memorize the exact set of settled neighbors of $w$.
Instead, it only remembers the maximum $i$ such that 
the first $i$ neighbors of $w$ are settled.
In contrast, $\rootopt$ allows each settler $\psi(w)$ to remember all settled neighbors of $w$ using $O(\Delta)$ bits,
which significantly helps to eliminate an $O(\log \tau)$ factor from the time complexity.
However, somewhat surprisingly,
both the design of the new algorithm and the analysis of its execution time are non-trivial.

Below, we outline the modifications made to $\rootal$ to obtain $\rootopt$.
We expect that most readers will grasp the behavior of $\rootopt$ simply by reviewing the following key differences,
while we provide the pseudocode in Algorithm \ref{al:root_optimal}.

\begin{itemize}
\item
In $\rootopt$, each settler $\psi(w)$ maintains an array variable $\psi(w).\checked$ of size $\delta_w$. Each element of $\psi(w).\checked[i]$ takes a value from the set $\{0, 1, \bot\}$. The assignment $\psi(w).\checked[i]=0$ (respectively, $1$) indicates that the neighbor $N(w,i)$ is unsettled (respectively, settled). The value $\bot$ is utilized exclusively during the probing process, meaning that the neighbor $N(w,i)$ has yet to be checked for its settled status. For any $c \in \{0,1,\bot\}$, we define $U_{c}(u)=\{N(u,i) \mid i \in [0,\delta_w-1], \psi(w).\checked[i]=c\}$.

\item Consider that the unique leader $\amax$ invokes the probing process $\Probe()$ at a node $w$.
(Remember that $\amax$ is the unique leader that has the maximum identifier at the beginning of the execution.)
In $\rootal$, the probing process terminates as soon as
any agent finds an unsettled neighbor. However, in $\rootopt$, the process only ends when all of $w$'s neighbors are probed or when at least $\ell$ unsettled neighbors are found, where $\ell$ is the number of explorers.
In the former case, the probing process is now complete: $\ubot(w)$ is empty, and $\uemp(w)$ equals the set of unsettled neighbors of $w$.
In the latter case, the explorers go to distinct $\ell$ unsettled nodes 
and settle there, thereby achieving dispersion.
\item
Consider an agent $a$ making a round trip $w \to u \to w$ during $\Probe()$, where $u = N(w,p)$ for some $p \in [0,\delta_w-1]$.
If $a$ does not encounter a settler at $u$, it simply sets $\psi(w).\checked[i]$ to $0$. On the other hand, if a settler is found at $u$, $a$ sets $\psi(w).\checked[i]$ to $1$ and additionally sets $\psi(u).\checked[q]$ to $1$, where $q \in [0,\delta_u-1]$ is the port number such that $w = N(u,q)$.
This modification ensures that $\uemp(u)$ remains equal to the set of unsettled neighbors of $u$ when 
the explorers go back to $u$.
\item
In $\rootal$, the leader $\amax$ invokes $\Probe()$ after each forward or backward move. However, in $\rootopt$, $\amax$ only invokes $\Probe()$ after making a forward move. This change does not compromise the correctness of $\rootopt$ because when $\amax$ makes a backward move to a node $w$, $\psi(w)$ accurately remembers its unsettled neighbors due to the modification mentioned earlier.
\end{itemize}

One might think that the probing process $\Probe()$ in $\rootopt$ could take longer time than in $\rootal$, as it only finishes after all neighbors of the current location have been probed or after finding $\ell$ unsettled neighbors, where $\ell$ is the number of explorers. Particularly, there seems to be a concern that during the probing at a node $w$, the number of agents, excluding $\psi(w)$, may not always double: this event occurs when some agents discover an unsettled neighbor. 
Despite that, we deny this conjecture at least asymptotically, that is, we have the following lemma.
\begin{lemma}
\label{lem:fastprobe}
Assume that $\amax$ invokes $\Probe()$ at node $w$ during the execution of $\rootopt$, and exactly $\ell$ explorers including $\amax$ exists at the time. Then, $\Probe()$ finishes within $O(1) + \max(0, 2\lceil \log \tau - \log \ell \rceil)$ time.
\end{lemma}
\begin{proof}
If $\ell \ge \tau$, we have $\ell \ge \min(\Delta,k) = \Delta$ because $\ell < k$.
Then, the lemma trivially holds: $\Probe()$ finishes in a constant time.
Thus, we consider the case $\ell < \tau$.
Let $t$ be the time step at which $\amax$ invokes $\Probe()$ at a node $w$,
and let $z = \lceil \log \tau - \log \ell \rceil +1$.
It suffices to show that 
$\ubot(w,t') = \emptyset$ or $|\uemp(w,t')| \ge \ell$ holds
for some $t' \in [t,t+2z+O(1)]$.
Assume for contradiction that this does not hold.
For any $r \in [0,z]$,
we define $f(r) = X_r \cdot 2^{r-1} +Y_r$,
where $X_r = |\uemp(w,t+2r)|$ and $Y_r = |A(w,t+2r)\setminus \{\psi(w)\}|$.
By definition, $f(0)=0\cdot 2^{0-1} + \ell = \ell$.
Under the above assumption, for any $r=0,1,\dots,z-1$, 
the agents in $A(w,t+2r)$ move to distinct neighbors in $\ubot(w,t+2r)$
in time step $t+2r$,
and bring back all settlers they find, at most one for each neighbor,
in time step $t+2r+1$.
Let $\alpha$ be the number of those settlers
\ie $\alpha = Y_{r+1}-Y_r$. 
Note that $X_{r+1}-X_r =Y_r-\alpha$ holds here.
Then, irrespective of $\alpha$,
we obtain
\begin{align*}
f(r+1) &= X_{r+1}\cdot 2^{r} + Y_{r+1} = (X_r+Y_r-\alpha)\cdot 2^{r} + Y_r+\alpha\\
& = X_r\cdot 2^r + (2^r-1)(Y_r-\alpha)+2\cdot Y_r
\ge  2(X_r \cdot 2^{r-1}+Y_r) = 2f(r),
\end{align*}
where we use $2^r \ge 1$ and $Y_r - \alpha = X_{r+1}-X_r \ge 0$ 
in the above inequality.
Therefore, we have $f(z) \ge \ell \cdot 2^z$,
whereas
we have assumed (for contradiction) that $X_z = |\uemp(w,t+2z)|< \ell$,
thus $f(z)=X_z\cdot 2^{z-1}+Y_z \le (\ell - 1)2^{z-1}+Y_z$ holds.
This yields $|A(w,t+2z)\setminus \{\psi(w)\}|=Y_z \ge \ell \cdot 2^z - (\ell - 1)2^{z-1} = (\ell+1) 2^{z-1} \ge (\ell+1)\tau/\ell > \tau$.
Since $\tau = \min(\Delta,k)$, we have $\tau = \Delta$ or $\tau = k$.
In the former case, $Y_k > \Delta$ agents at $w$ are enough to visit all neighbors in $\ubot(w,t+2z)$
in time step $t+2z$, thus $\ubot(w,t+2z+2)=\emptyset$ holds, a contradiction. 
In the latter case, there are $|A(w,t+2z)| \ge k+2$ agents in $w$ at time step $t+2z$, a contradiction. 
Therefore, $\ubot(w,t') =  \emptyset$ or $|\uemp(w,t')|\ge \ell$ holds
at time step $t' \le t+2z+2$.
\end{proof}

\begin{theorem}
\label{theorem:rootopt}
In the rooted setting,
algorithm $\rootopt$ 
solves the dispersion problem within $O(k)$ time using $O(\Delta+\log k)$ bits of space per agent.
\end{theorem}
\begin{proof}
The unique leader $\amax$ invokes $\Probe()$ only when it settles an agent, except for when $\amax$ itself becomes settled. Therefore, $\amax$ invokes $\Probe$ exactly $k-1$ times, with precisely $k-i$ explorers present at the $i$-th invocation. By Lemma \ref{lem:fastprobe}, the total number of steps required for the $k-1$ executions of $\Probe()$ is at most
$\sum_{\ell= 1}^{k-1} (\log k - \log \ell+O(1)) = k \log k - (\log (k!)- \log k) + O(k) = O(k)$,
where we apply Stirling's formula, \ie $\log (k!) = k \log k - k +O(\log k)$.
As demonstrated in Section \ref{sec:rooted}, both forward and backward moves also require a total time of $O(k)$. Thus, $\rootopt$ completes in $O(k)$ time. Regarding space complexity, the array variable $\checked$ is the primary factor, needing $O(\Delta)$ bits per agent. Other variables require only $O(\log \Delta)$ bits.
\end{proof}

\section{Concluding Remarks}
\label{sec:discussion}
 In this paper, we introduced a novel technique HEO-DFS and presented near time-optimal dispersion algorithms based on this techniques both for the rooted setting and the general setting. Both the algorithms require only $O(\log (k+\Delta)$ bits per agent. Moreover, 
 we give time-optimal algorithms that solves dispersion in the rooted setting at the cost of increasing the memory space to $O(\Delta + \log k)$ bits per agent. 

It is worth mentioning that while HEO-DFS does not function in a \emph{fully} asynchronous model, where movement between two nodes may require an unbounded period, it does not require a \emph{fully} synchronous model in the rooted setting. Specifically, $\rootal$ and $\rootopt$ can operate under an asynchronous scheduler if every movement of agents between nodes is \emph{atomic}, \ie each agent is always located at a node and never on an edge at any time step. Under this scheduler, after the probing process is completed at a node $v$, unsettled agents can wait for all helping settlers to leave $v$ before they themselves depart. Since every movement is atomic, when unsettled agents visit the home node of one of these settlers, the settler has already returned. Thus, $\rootal$ and $\rootopt$ functions under any fair scheduler that guarantees every movement is atomic. However, this move-atomicity is not sufficient for $\general$, because this algorithm, designed for the general setting, differentiates the moving speeds of agents based on their roles -- leader, strong zombie, or weak zombie -- which inherently requires a fully synchronous scheduler.


\clearpage

\bibliography{dispersion}

\clearpage

\appendix


\section{How to Terminate \texorpdfstring{$\rootal$}{RootedDispersion} and \texorpdfstring{$\rootopt$}{RootOptimal}}
\label{sec:termination}
The termination can be easily implemented in a common way both for $\rootal$ and $\rootopt$.

The last settled agent $\alast$ recognizes the achievement of dispersion immediately after all agents are settled at distinct nodes. In $\rootopt$, the last settled agent $\alast$ may not be uniquely determined when the unsettled agents at node $w$ detect that $|\uemp(w)| \ge \ell$, where $\ell$ is the number of those agents, and they move to and settle at $\ell$ distinct nodes (line 33 in Algorithm \ref{al:root_optimal}). In this case, those $\ell$ agents immediately terminate and are not used for the subsequent termination detection. Instead, we define $\alast$ as $\psi(w)$, who settled at $w$ before those $\ell$ agents.

Roughly speaking, $\alast$ then relays the fact that dispersion is achieved to all remaining agents via edges in the DFS tree, say $T$, that the agents have constructed so far. For any settled node $v \neq w = \lctn(\amax)$, we define $f(v)$ as the next node of $v$ in the path leading to $w$ in tree $T$. Formally, we define $f(v) = u_1$, where $p = u_0, u_1, \dots, u_s\ (v = u_0,\ w = u_s)$ is the unique path between $v$ and $w$.
Like the time slots in $\general$, we divide time steps into three time slots: 1, 2, and 3. Specifically, we classify time step $t$ as slot $i \in \{1, 2, 3\}$ if $(t \bmod 3) + 1 = i$. In every time slot 1, each agent moves or stays according to the original algorithms, \ie Algorithms \ref{al:rootal} and \ref{al:root_optimal} for $\rootal$ and $\rootopt$, respectively. In slot 2, each settled agent $b$ with home $\lctn(b)$ moves to $f(\lctn(b))$, and in slot 3, it goes back to $\lctn(b)$. In slots 2 and 3, no unsettled agent moves.
After $\alast$ detects the achievement of dispersion at time step $t$, which must be slot 1, it does not terminate immediately. It waits for other settlers coming to $w$ in the next slot 2, tells them about the achievement of dispersion in slot 3, and halts in the next slot 1. The settlers who newly learn this fact go back to their homes in slot 3 and execute the same procedure: wait until the next slot 3, inform the arriving settlers about the achievement, and halt in the next slot 1. Since the diameter of the DFS tree $T$ is at most $k - 1$, this termination procedure only requires $O(k)$ time.

Each settler $b$ with home $v = \lctn(b)$ must recognize which ports $0, 1, \dots, \delta_{v} - 1$ lead to node $f(v)$. This recognition can be easily implemented. For $\rootal$, settler $b$ knows $f(v) = N(v, b.\nxt)$ if $b.\nxt \neq \bot$, and $f(v) = N(v, b.\parent)$ otherwise. 
We slightly modify $\rootopt$ for this recognition:
(i) settler $b$ maintains a variable $b.\nxt$,
(ii) $b$ sets $b.\nxt \gets u$ in Line 6, and 
(iii) $b$ sets $b.\nxt \gets b.\parent$ in Line 11.
Then, $b$ always knows $f(v) = N(v, b.\nxt)$.

\end{document}